\documentclass[letterpaper, 10 pt, conference]{ieeeconf}  %

\IEEEoverridecommandlockouts                              %
\usepackage{amsmath,amssymb,amstext}

\usepackage{resizegather}

\usepackage{pgfplots} %
\pgfplotsset{compat=newest} 
\pgfplotsset{plot coordinates/math parser=false}

\usepackage{physics}
\usepackage{siunitx}

\usepackage{algorithm}
\usepackage{algorithmic}

\usepackage{xcolor}
\usepackage{mathtools}

\usepackage{hyperref}
\hypersetup{
	colorlinks=true,
	linkcolor=blue,
	filecolor=blue,      
	urlcolor=blue,
	citecolor= blue,
}

\let\labelindent\relax
\usepackage{enumitem}

\usepackage{algorithm}
\usepackage{algorithmic}

\newcommand{\defeq}{:=}
\newcommand{\change}[2]{#2}

\newcommand{\auxvar}{\tau_{\max}}

\newcommand{\lvar}{\Lambda}
\newcommand{\tpar}{\text{\normalfont par}}
\newcommand{\refer}{\text{\normalfont ref}}

\newcommand{\genC}{C(x(\tvar_j),\eta(\tvar_j),c)}
\newcommand{\windowC}{C(x(\tvar_j),\eta(\tvar_j),c)}%

\newcommand{\itilde}{1}

\newcommand{\svar}{t}
\newcommand{\tvar}{s}

\def\QEDclosed{\mbox{\rule[0pt]{1.3ex}{1.3ex}}} %

\def\qed{\QEDclosed} %

\makeatletter
\newcommand{\pushright}[1]{\ifmeasuring@#1\else\omit\hfill$\displaystyle#1$\fi\ignorespaces}
\newcommand{\pushleft}[1]{\ifmeasuring@#1\else\omit$\displaystyle#1$\hfill\fi\ignorespaces}
\makeatother

\newtheorem{defi}{Definition}
\newtheorem{theo}{Theorem}
\newtheorem{rema}{Remark}

\newtheorem{coro}{Corollary}

\newtheorem{asum}{Assumption}
\newtheorem{prop}{Proposition}

\title{\LARGE \bf
Dynamic self-triggered control for nonlinear systems based on hybrid Lyapunov functions
}

\author{Michael Hertneck and Frank Allg\"ower%
	\thanks{Funded by Deutsche
		Forschungsgemeinschaft (DFG, German Research Foundation) under Germany’s
		Excellence Strategy - EXC 2075 - 390740016 and under grant
		AL 316/13-2 - 285825138. We acknowledge the support by the Stuttgart
		Center for Simulation Science (SimTech).}%
	\thanks{The authors are with the University of Stuttgart, Institute for Systems Theory and Automatic Control,  70569 Stuttgart, Germany (email: $\{$hertneck, allgower$\}$@ist.uni-stuttgart.de).}
}

\begin{document}

\pubid{\begin{minipage}{\textwidth}\ \\[12pt] \copyright 2021 IEEE. Personal use of this material is permitted. Permission from IEEE must be obtained for all other uses, in any current or future media, including reprinting/republishing this material for advertising or promotional purposes, creating new collective works, for resale or redistribution to servers or lists, or reuse of any copyrighted component of this work in other works.\end{minipage}} 

\maketitle
\pagestyle{empty}

\begin{abstract}
	Self-triggered control (STC) is a well-established technique to reduce the amount of samples for sampled-data systems, and is hence particularly useful for Networked Control Systems. At each sampling instant, an STC mechanism determines not only an updated control input but also when the next sample should be taken. 
	In this paper, a dynamic STC mechanism for nonlinear systems is proposed. The mechanism incorporates a dynamic variable for determining the next sampling instant. Such a dynamic variable for the trigger decision has been proven to be a powerful tool for increasing sampling intervals in the closely related concept of event-triggered control, but was so far not exploited for STC. This gap is closed in this paper.
	For the proposed mechanism, the dynamic variable is chosen to be the filtered values of the Lyapunov function at past sampling instants. The next sampling instant is, based on the dynamic variable and on hybrid Lyapunov function techniques, chosen such that an average decrease of the Lyapunov function is ensured. 
	The proposed mechanism is illustrated with a numerical example from the literature. For this example, the obtained sampling intervals are significantly larger than for existing static STC mechanisms. This paper is the accepted version of \cite{hertneck21dynamic}, containing also proofs of the main results. 
\end{abstract}

\section{Introduction}
Implementing feedback laws using shared hardware or shared communication channels has recently received growing interest due to various applications in networked and embedded control systems.
In such setups, classical periodic sampling of feedback laws often leads to a waste of resources \cite{astroem1999comparison,arzen1999simple}. For Networked Control Systems (NCS), there exist well-known alternatives to periodic sampling based on state-dependent trigger mechanisms instead of time-dependent ones, as, e.g., event- and self-triggered control \cite{heemels2012introduction}. In event-triggered control (ETC), the trigger rule is monitored continuously and a transmission is triggered as soon as the trigger rule is violated. 
In self-triggered control (STC), in contrast, at each sampling instant the controller determines proactively when the next sample should be taken.
It has been demonstrated that STC can reduce the network load for NCS significantly \cite{mazo2009self,anta2010sample}. For linear systems, various STC approaches are available, see, e.g.,  \cite{heemels2012introduction,brunner2019event} and the references therein. 
For nonlinear systems, fewer results are available. In \cite{anta2010sample,delimpaltadakis2020isochronous,delimpaltadakis2020region} isochronous manifolds are employed to determine the next sampling instant. In \cite{benedetto2011digital,tiberi2013simple,theodosis2018self,proskurnikov2019lyapunov}, continuity properties of the considered systems are utilized to choose sampling instants.
All these works have in common that only the current state of the system is taken into account for determining the next sampling instant and no past information about the system behavior is used.

However, for event-triggered control, it has been shown that taking into account such information, e.g., using an averaged version of a purely state-dependent trigger rule, can reduce the required amount of samples significantly \cite{girard2015dynamic}. It therefore appears promising to take the past system behavior into account also for STC.

\pubidadjcol 
To close this gap, we propose in this paper a novel \textit{dynamic STC} mechanism that uses a dynamic variable to increase sampling intervals. The dynamic variable is chosen as an averaged version of a Lyapunov function and thus encodes the past system behavior. The update of the dynamic variable based on sampled values of the Lyapunov function can be described by a finite-impulse-response (FIR) low-pass filter. 
The next sampling instant is chosen at the current one as the largest time for which the value of the Lyapunov function at the next sampling instant is bounded by the value of the dynamic variable. To achieve this, hybrid Lyapunov techniques originally based on \cite{carnevale2007lyapunov,nesic2009explicit} and adapted from \cite{hertneck20stability} are employed to bound the future evolution of the Lyapunov function and to determine the sampling instants. Combining the dynamics of the dynamic variable and the bound on the Lyapunov function, an average decrease of the Lyapunov function can be guaranteed. 
From this, we derive stability guarantees for the dynamic STC mechanism.
 Moreover, the proposed dynamic STC mechanism allows a trade-off between the computational complexity and the conservativeness when determining the next sampling instant. 
 We illustrate the proposed dynamic STC mechanism with a numerical example from the literature and demonstrate that a significant reduction of the amount of samples in comparison to existing static STC mechanisms can be obtained.

 The remainder of this paper is structured as follows. In Section~\ref{sec_setup}, we specify the considered setup using a hybrid system description. Some preliminaries are discussed in Section~\ref{sec_pre}. In Section~\ref{sec_main}, details of the proposed STC mechanism are presented and stability guarantees are derived.
  A numerical example is given in Section~\ref{sec_ex}. Section~\ref{sec_sum} summarizes the paper and gives an outlook on possible future research directions. 

\subsection*{Notation and definitions}
The nonnegative real numbers are denoted by  $\mathbb{R}_{\geq 0} $. The natural numbers are denoted by $\mathbb{N}$, and we define $\mathbb{N}_0:=\mathbb{N}\cup  \left\lbrace 0 \right\rbrace $. The smallest eigenvalue of a  real symmetric matrix $P$ is denoted by $\lambda_{\min}(P)$. A continuous function $\alpha: \mathbb{R}_{\geq 0} \rightarrow \mathbb{R}_{\geq 0}$ is a class $ \mathcal{K}$ function if it is strictly increasing and $\alpha(0) = 0$. It is a class $\mathcal{K}_\infty$ function if it is of class $\mathcal{K}$ and it is unbounded. A continuous function $\beta:\mathbb{R}_{\geq 0}\times \mathbb{R}_{\geq 0} \rightarrow \mathbb{R}_{\geq 0}$ is a class $\mathcal{K}\mathcal{L}$ function, if $\beta(\cdot,t)$ is of class $\mathcal{K}$ for each $t \geq 0$ and $\beta(q,\cdot)$ is nonincreasing and satisfies $\lim\limits_{t \rightarrow \infty} \beta(q,t) = 0$ for each $q \geq 0$. A function $\beta:\mathbb{R}_{\geq 0}\times \mathbb{R}_{\geq 0} \times \mathbb{R}_{\geq 0} \rightarrow \mathbb{R}_{\geq 0}$ is a class $\mathcal{K}\mathcal{L}\mathcal{L}$ function if for each $r \geq 0$, $\beta(\cdot,r,\cdot)$ and $\beta(\cdot,\cdot,r)$ belong to class $\mathcal{K}\mathcal{L}$.

We use \cite[Definitions 1-3]{carnevale2007lyapunov}, that are originally taken from \cite{goebel2006solutions}, to characterize a hybrid model of the considered NCS and corresponding hybrid time domains, trajectories and solutions.
\section{Setup}
\label{sec_setup}
We consider a setup where a continuous plant and a dynamic controller exchange information only at discrete sampling instants. A reason for the discrete sampling may be, e.g., that plant and controller are connected through a communication network.
The plant is described by
\begin{equation}
	\label{eq_plant}
	\dot{x}_p = f_p(x_p,\hat{u}),
\end{equation}
where $x_p(t) \in \mathbb{R}^{n_x}$ is the plant state with initial condition  $x_p(0) = x_{p,0}$ and $\hat{u}(t) \in \mathbb{R}^{n_u}$ is the last input that has been received by the plant. The input is generated by the controller
\begin{equation}
	\label{eq_controller}
	\begin{split}
	\dot{x}_c = f_c(x_c,\hat{x}_p)\\
	u = g_c(x_c,x_p),
	\end{split}
\end{equation}
where $x_c(t)\in\mathbb{R}^{n_c}$ is the state of the controller with some initial condition $x_c(0) = x_{c,0}$ and $\hat{x}_p$ is the last plant state that has been received by the controller.  Note that this controller model admits also static state-feedback controllers. 

The plant is sampled at sampling instants $\svar_j,j\in\mathbb{N}_0$, that are determined by a sampling mechanism to be specified later. 
At each sampling instant, $\hat{x}_p$ and $\hat{u}$ are updated based on the current values of $x_p$ and $u$, i.e., $\hat{x}_p(\svar_j) = x_p(\svar_j)$ and $\hat{u}(\svar_j) = u(\svar_j)$. Further, we denote by $\hat{x}_c$ the controller state associated to the last sampling instant, i.e, $\hat{x}_c(\svar_j) = x_c(\svar_j)$.   Between sampling instants, we assume that $\hat{x}_p$, $\hat{u}$ and $\hat{x}_c$ are kept constant, which resembles a zero-order-hold (ZOH) scenario. 

We introduce the sampling-induced error $e=\left[e_{x_p}^\top,e_{x_c}^\top\right]^\top \coloneqq\left[(\hat{x}_p-x_p)^\top,(\hat{x}_c-x_c)^\top\right]^\top$ and the combined state $x \coloneqq \left[x_p^\top,x_c^\top\right]^\top$. Note that $e(t) \in \mathbb{R}^{n_e}$ and $x(t) \in \mathbb{R}^{n_x}$ for $n_x = n_e \coloneqq n_{x_p} +n_{x_c}$. 

We consider in this paper a dynamic STC mechanism. Similar as for static STC mechanisms, the considered mechanism determines at the current sampling instant  $\svar_j$ the next sampling instant $\svar_{j+1}$. In addition to the current states of plant and controller, the dynamic mechanism uses an internal state $\eta(t)\in\mathbb{R}^{n_\eta}$ which allows to incorporate the past system behavior in the trigger decision. 
The dynamic sampling mechanism can then be described by 
$\svar_{j+1} \coloneqq \svar_j + \Gamma(x(\svar_j),\eta(\svar_j)),$
where $\Gamma:\mathbb{R}^{n_x}\times\mathbb{R}^{n_\eta} \rightarrow \left[\svar_{\min},\infty\right)$ for some $\svar_{\min} > 0$. We will provide later an explicit value for $t_{\min}$ for the proposed mechanism. 

For reasons of practical implementability, the dynamic variable $\eta$ is only updated at sampling instants based on its current value and on current state information, and remains constant in between sampling instants. Thus, $\eta$ evolves according to 
	$	\eta(\svar_{j+1}) = S(\eta(\svar_j),x(\svar_j))$
for some $\eta(0)$, where $S:\mathbb{R}^{n_\eta}\times\mathbb{R}^{n_x} \rightarrow \mathbb{R}^{n_\eta}.$

We introduce the timer variable $\tau$ which keeps track of the elapsed time since the last sampling instant and the auxiliary variable $\auxvar$ which encodes the next sampling interval. With this, we can model the overall networked control system as the hybrid system
\begin{equation}
\label{eq_sys_hyb}
	\begin{cases}
		\dot{\xi} = F(\xi), & \xi \in C,\\
		\xi^+ = G(\xi), & \xi \in D,
	\end{cases}
\end{equation}
with $\xi \coloneqq \left[x^\top,e^\top,\eta^\top,\tau,\auxvar\right]^\top,$
\begin{equation*}
	F(\xi) \coloneqq \left(f(x,e)^\top,g(x,e)^\top,0,1,0\right)^\top,
\end{equation*}
with
\begin{align*}
	f(x,e) = \begin{bmatrix}
		f_p(x_p,g_c(x_c+e_{x_c},x_p+e_{x_p}))\\
		f_c(x_c,x_p+e_{x_p}) 
	\end{bmatrix}	
\end{align*}
and $g(x,e) = -f(x,e)$,
\begin{equation*}
	G(\xi) \coloneqq \left(x^\top,0,S(\eta,x)^\top,0,\Gamma(x,\eta)\right)^\top,
\end{equation*}
and with 
\begin{equation}
	\begin{split}
	C := \left\lbrace \xi \in \mathbb{R}^{n_x+n_e+n_\eta+2} | \tau \leq \auxvar \right\rbrace\\
	D := \left\lbrace \xi \in \mathbb{R}^{n_x+n_e+n_\eta+2} | \tau = \auxvar \right\rbrace.
	\end{split}
\end{equation}

Since sampling instants are modeled as jumps of the hybrid system~\eqref{eq_sys_hyb}, the sampling sequence $\svar_j, j \in \mathbb{N}_0$ corresponds exactly to the time indices when \eqref{eq_sys_hyb} jumps. We thus describe by $\tvar_j \coloneqq  (\svar_j,j)$ the hybrid time before the sampling at time $\svar_j$ and by $\tvar_j^+ \coloneqq (\svar_j,j+1)$ the hybrid time directly after the sampling at time $\svar_j$\change{}{, i.e., the hybrid time instants just before and just after $t_j$}. 

We assume that the self-triggered sampling mechanism is executed at the initial time $t_0 = 0$, i.e., there is a jump at $t = 0$ and we have that $\xi(0,1) = G(\xi(0,0))$. This can be incorporated into the hybrid system \eqref{eq_sys_hyb} by restricting the initial conditions for $\tau(0,0)$ and $\auxvar(0,0)$ to $\auxvar(0,0) =\tau(0,0)$. Note that without this assumption, the first sampling instant might be not well-defined.

Subsequently, our goal will be to design functions $\Gamma$ and $S$ that ensure asymptotic stability of the origin of \eqref{eq_sys_hyb} for a region of attraction $\mathcal{R}$ according to the following definition, which is adapted to our setup from \cite[{Definition~4}]{carnevale2007lyapunov}.

\begin{defi}
	\label{def_stab}
	For the hybrid system~\eqref{eq_sys_hyb}, the set $\left\lbrace \left(x,e,\eta,\tau,\auxvar\right): x = 0, e= 0, \eta = 0 \right\rbrace$ is asymptotically stable with region of attraction $\mathcal{R} \subseteq \mathbb{R}^{n_x}$, if there exists $\beta \in \mathcal{K}\mathcal{L}\mathcal{L}$ such that, for each initial condition $x(0,0) \in \mathcal{R}$, $\eta(0,0) \in \mathbb{R}^{n_\eta}$, $e(0,0) \in \mathbb{R}^{n_e}$, $\tau(0,0) \in \mathbb{R}_{\geq 0}$ and $\auxvar(0,0) = \tau(0,0) $, and each corresponding solution
	\begin{equation}
	\label{eq_stab_bound}
	\norm{\begin{bmatrix}
		x(t,j)\\
		e(t,j)\\
		\eta(t,j)
		\end{bmatrix}} \leq \beta\left(\norm{\begin{bmatrix}
		x(0,0)\\
		e(0,0)\\
		\eta(0,0)
		\end{bmatrix}},t, j\right)
	\end{equation}
	for all $(t,j)$ in the solutions domain. 
\end{defi} 
While local results are typically more relevant for nonlinear systems, our results are also applicable to obtain global guarantees. In this case, we use the following definition.
\begin{defi}
	If the hybrid system~\eqref{eq_sys_hyb} satisfies Definition~\ref{def_stab} with $\mathcal{R} = \mathbb{R}^{n_x}$, then the set $\left\lbrace \left(x,e,\eta,\tau,\auxvar\right): x = 0, e= 0, \eta = 0 \right\rbrace$ is uniformly globally asymptotically stable.
\end{defi}
 \section{Preliminaries}
 \label{sec_pre}
 In this section, we present an assumption that we make on the dynamics of the hybrid system~\eqref{eq_sys_hyb} and state a preliminary result that we will exploit later to determine sampling instants. 
  We make the following assumption that is based on \cite[Assumption~1]{carnevale2007lyapunov}.
  \begin{asum}
  	\label{asum_hybrid_lyap}
  	Consider some sets $\mathcal{X} \subseteq \mathbb{R}^{n_x}$ and $\mathcal{E} \subseteq \mathbb{R}^{n_e}$. There exist a locally Lipschitz function $W:\mathcal{E} \rightarrow \mathbb{R}_{\geq0}$, a locally Lipschitz function $V:\mathcal{X} \rightarrow \mathbb{R}_{\geq0}$, a continuous function $H:\mathcal{X}\times\mathcal{E} \rightarrow \mathbb{R}_{\geq0}$, constants $L, \gamma\in \mathbb{R}_{>0}$, $\epsilon\in \mathbb{R}$, and  $\underline{\alpha}_W$, $\overline{\alpha}_W, \underline{\alpha}_V, \overline{\alpha}_V \in \mathcal{K}_\infty$  such that for all $e\in\mathcal{E}$,
  	\begin{equation}
  	\label{eq_w_bound}
  	\underline{\alpha}_W(\norm{e}) \leq W(e) \leq \overline{\alpha}_W(\norm{e}),
  	\end{equation}
  	for all $x \in\mathcal{X}$,
  	\begin{equation}
  	\label{eq_V_bound_K}
  	\underline{\alpha}_V(\norm{x}) \leq V(x) \leq \overline{\alpha}_V(\norm{x}),
  	\end{equation}
  	and for all  $x \in \mathcal{X} $ and almost all $e\in\mathcal{E},$ 
  	\begin{equation}
  	\left\langle \frac{\partial W(e)}{\partial e},g(x,e)\right\rangle \leq L W(e) + H(x,e).  \label{eq_w_est}
  	\end{equation}
  	Moreover, for all $e \in \mathcal{E}$ and almost all $x \in \mathcal{X}$,  
  	\begin{equation}
  	\left\langle \nabla V(x),f(x,e) \right\rangle
  	\leq - \epsilon V(x) -H^2(x,e) + \gamma^2 W^2(e).	\label{eq_v_desc_hybrid}
  	\end{equation}
  	\end{asum}

  	\vspace{5pt}
  	Note that there are some differences between Assumption~\ref{asum_hybrid_lyap} and \cite[Assumption~1]{carnevale2007lyapunov}. Firstly, Assumption~\ref{asum_hybrid_lyap} is stated in a local fashion. A global version similar to \cite[Assumption~1]{carnevale2007lyapunov} can be obtained for $\mathcal{X} = \mathbb{R}^{n_x}$ and $\mathcal{E} = \mathbb{R}^{n_e}$. Secondly, we restrict ourselves to systems that satisfy the exponential input-to-state stability (ISS) condition \eqref{eq_v_desc_hybrid} instead of considering the more general ISS condition from \cite[Assumption~1]{carnevale2007lyapunov} (i.e., we replace $\varrho(\norm{x})$ by $\epsilon V(x)$ in \eqref{eq_v_desc_hybrid}). This simplification is made to streamline presentation and can in principle be relaxed as we will discuss at the end of Section~\ref{sec_main}. 
  	Thirdly, we use the more general function $H(x,e)$ instead of restricting ourselves to $H(x)$, which often admits smaller values for $\gamma$ and $L$, as it has been demonstrated in \cite{hertneck20simple}. \change{}{Moreover, this modification allows us for $W(e) = \norm{e}$ to chose $H(x,e) = \norm{g(x,e)}$ and $L$ arbitrarily small and thus yields a systematic procedure to verify \eqref{eq_w_est} in this case.}
  	  	 
  	 It is important to note here that Assumption~\ref{asum_hybrid_lyap} can  typically be satisfied for various different choices of the parameters $\epsilon,\gamma$ and $L$. If one can find one set of parameters, then one will typically also be able to find many different sets of parameters, e.g., by varying $\epsilon$. \change{}{Approaches for verifying Assumption~\ref{asum_hybrid_lyap} for given $\epsilon$ can, e.g., be found in \cite{hertneck20simple}.}  
  	 
  	The proposed dynamic STC mechanism will be based on a bound on the evolution of $V$ between two sampling instants, which has already been presented in \cite{hertneck20stability}. This bound employs the function
  	\begin{equation}
  	T_{\max}(\gamma,\lvar ) \defeq \begin{cases}\vspace{1mm}
  	\frac{1}{\lvar r} \mathrm{arctan}(r) & \gamma > \lvar \\ \vspace{1mm}
  	\frac{1}{\lvar } & \gamma = \lvar \\
  	\frac{1}{\lvar r} \mathrm{arctanh}(r) &\gamma < \lvar 
  	\end{cases}
  	\end{equation}
  	where
  	\begin{equation}
  	\label{eq_def_r}
  	r\defeq\sqrt{\abs{
  			\left(\frac{\gamma}{\lvar }\right)^2-1}}.
  			\end{equation} 
  			This function was originally used in  \cite{nesic2009explicit}. \change{}{In Figure~\ref{fig_surf}, an overview over the function $T_{\max}(\gamma,\Lambda)$ is given in the form of a surface plot. It can be seen that the function gets larger as $\gamma$ and $\Lambda$ get smaller.}
  			
  			\begin{figure}
  				\resizebox{\columnwidth}{!}{\includegraphics{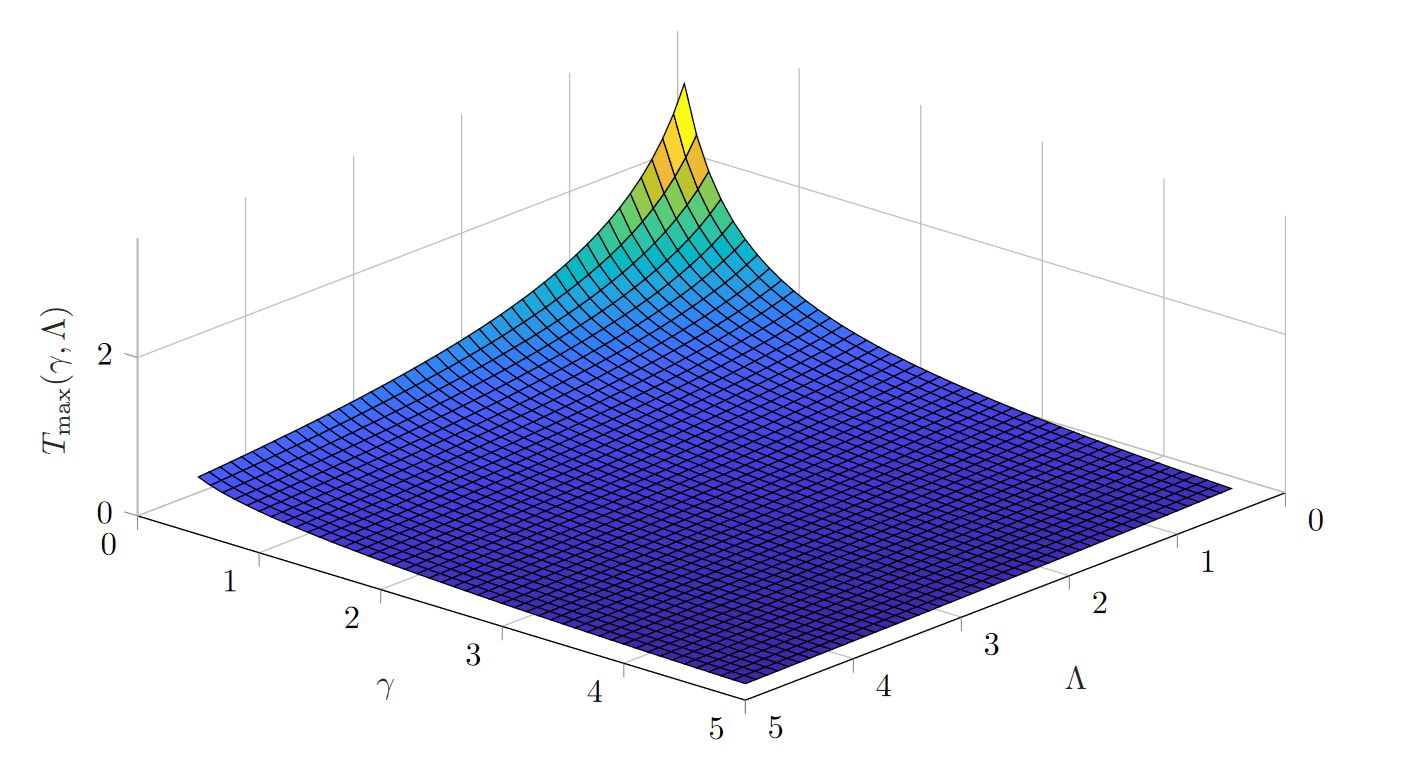}}
  				\caption{Surface plot of $T_{\max}(\gamma,\Lambda)$.}
  				\label{fig_surf}
  			\end{figure}
  			
  			 \change{In particular, we}{We } adapt from \cite[Proposition~12]{hertneck20stability} the following result.
  			
	\begin{prop}
		\label{prop_hybrid}
		Consider system \eqref{eq_sys_hyb} at sampling instant $\tvar_j^+$ for some $j \in \mathbb{N}_0$.  Let Assumption~\ref{asum_hybrid_lyap} hold for $\mathcal{X} \supseteq \mathcal{R} \coloneqq \left\lbrace x|V(x) \leq c \right\rbrace$ and\footnote{Note that $\left\lbrace e = \hat{x}-x|\hat{x}\in \mathcal{R},x\in\mathcal{R}\right\rbrace$ is the Minkowski sum of $\mathcal{R}$ and $-\mathcal{R}$.} $\mathcal{E} \supseteq \left\lbrace e = \hat{x}-x|\hat{x}\in \mathcal{R},x\in\mathcal{R}\right\rbrace$ for some $\gamma, \epsilon, L$ and $c \geq \max\left\lbrace 1, e^{\max\left\lbrace -\epsilon,2(L-\lvar ) \right\rbrace(\auxvar(s_j^+))} \right\rbrace V(x(\tvar_j^+))$. 
		
		Moreover, let $0 < \auxvar(\tvar_j^+) < T_{\max} (\gamma,\lvar )$  for $\lvar  >0$.	
		Consider 
		\begin{equation}
		\label{eq_def_u}
		U(\xi) = V(x)+\gamma \phi(\tau) W^2(e),
		\end{equation}
		where 	 $\phi : [0,\auxvar(\tvar_j^+)] \rightarrow \mathbb{R}$ is the solution to
		\begin{equation}
			\label{eq_def_phi}
			\dot{\phi} = -2\lvar \phi-\gamma(\phi^2+1),~ \phi(0) = \lambda^{-1}
		\end{equation}
		for some sufficiently small $\lambda \change{> 0}{\in \left(0,1\right)}$.
		Then, for all $\svar_j \leq t \leq \svar_j+\auxvar(\tvar_j^+)$, it holds that
		\begin{equation}
		\change{}{V(x(t))\leq} U(\xi(\svar,j+1)) 	\leq e^{\max\left\lbrace -\epsilon,2(L-\lvar ) \right\rbrace(t-\svar_j)} V(x(\tvar_j^+)). \label{eq_prop_hybrid1}
		\end{equation}
		\end{prop}
\begin{proof}
	See Appendix~\ref{app_hyb}.
\end{proof}

 \section{Dynamic STC with stability guarantees}
 \label{sec_main}
 In this section, we will present the details of the proposed dynamic STC mechanism and derive stability guarantees. First, we describe how the dynamic variable is used for the proposed dynamic STC mechanism. Then, we discuss how the results on hybrid Lyapunov functions from the previous section can be exploited to determine the sampling instants. Finally, we guarantee asymptotic stability for the proposed mechanism.
 \subsection{General idea for the dynamic STC mechanism}
 In this subsection, we will describe how the dynamic variable is used by the proposed dynamic STC mechanism. We are interested in guaranteeing stability of the set $\left\lbrace \left(x,e,\eta,\tau,\auxvar\right): x = 0, e= 0, \eta = 0 \right\rbrace$ for a region of attraction described by a level set $\mathcal{R} = \left\lbrace x|V(x) \leq c \right\rbrace$ of a Lyapunov function $V$ for some chosen $c > 0$. \change{}{The constant $c$ hereby needs to be chosen such, that the region of interest of the state space is covered by $\mathcal{R}$. } 
 
 \change{The idea \change{to achieve this}{for guaranteeing stability } is hereby to use Proposition~\ref{prop_hybrid} and the resulting bound on $V$ at sampling instant $\tvar_j$ to determine  the next sampling instant $\tvar_{j+1}$.}{}  \change{}{The idea for the dynamic STC mechanism is, to chose if possible $\tvar_{j+1}$ such that }\change{If possible, $\tvar_{j+1}$ is chosen such that }{}  $V(x(\tvar_{j+1}))$ does not exceed a discounted average of the past $m$ values of $V$ from previous sampling instants for some chosen $m > 0$, i.e., $\tvar_{j+1}$ is chosen such that  \begin{equation}
 \label{eq_idea_window}
 V(x(\tvar_{j+1})) \leq e^{-\epsilon_\refer(\svar_{j+1} - \svar_j)}\frac{1}{m} \sum_{k=j-m+1}^{j} V(x(\tvar_k)).
 \end{equation}
  Here, a tunable constant $\epsilon_\refer > 0$ is used to guarantee a certain decay of $V(x)$ in comparison to the past average\change{, which is important for obtaining stability guarantees for the proposed STC mechanism.}{, thus allowing us to ensure a decrease of the right hand side of \eqref{eq_idea_window} over time. If \eqref{eq_idea_window} cannot be guaranteed at a sampling instant, then the dynamic STC mechanism will use a fall-back strategy to be specified later. Both together will lead to stability guarantees for the dynamic STC mechanism.}

 \change{For implementing the average, the}{Next, we will explain how sufficient conditions for \eqref{eq_idea_window} can be implemented in our setup. The }  dynamic variable $\eta$ can be used as a storage for the past values of $V(x)$. To do so, we can choose $n_\eta = m-1$ and define the update rule \change{}{for the dynamic variable } as
 \begin{equation}
 \label{eq_S_window}
 S(\eta,x) = \begin{pmatrix}
 \eta_2&
 \hdots&
 \eta_{m-1}&
 V(x)
 \end{pmatrix}^\top.
 \end{equation}
 Observe that for this choice of $S(\eta,x)$, at time $\tvar_j$, if $j \geq m$, then it holds that $ \frac{1}{m} \sum_{k=j-m+1}^{j} V(x(\tvar_k)) =\frac{1}{m} \left(V(x(\tvar_j))+\sum_{k=1}^{m-1} \eta_k(\tvar_j)\right)$.
 
\change{ To be able}{To enable the dynamic STC mechanism } to use Proposition~\ref{prop_hybrid}, we will employ subsequently Assumption~\ref{asum_hybrid_lyap}  for $\mathcal{X} \supseteq \mathcal{R}$ and $\mathcal{E} \supseteq \left\lbrace e = \hat{x}-x|\hat{x}\in \mathcal{R},x\in\mathcal{R}\right\rbrace$. The STC mechanism therefore needs to ensure that $x(t,j) \in \mathcal{R}$ and $e(t,j) \in \mathcal{E}$ for all $(t,j)$. This will be achieved by the additional condition
 \begin{equation}
 \label{eq_V_c}
 	V(x(s_{j+1})) \leq c
 \end{equation}
 that has to be ensured by the dynamic STC mechanism. We can now summarize \eqref{eq_idea_window} and \eqref{eq_V_c} as
   \begin{equation}
   \label{eq_idea_window2}
   V(x(\tvar_{j+1})) \leq e^{-\epsilon_\refer(\svar_{j+1} - \svar_j)}  \windowC,
   \end{equation}
   where 
   \begin{align*}
   	&\windowC\\ 
   	\coloneqq& \min\left\lbrace c, \frac{1}{m} \left(V(x(\tvar_j))+\sum_{k=1}^{m-1} \eta_k(\tvar_j)\right) \right\rbrace.
   \end{align*}
\change{}{Thus, if the dynamic STC mechanism ensures \eqref{eq_idea_window2}, then it ensures \eqref{eq_idea_window} and in addition that Proposition~\ref{prop_hybrid} can be used. }
  For $j < m$, the value of $ \windowC$ may depend on the initial condition for $\eta$,  which can typically be chosen by the user to tune the initial behavior of the STC mechanism and does not influence stability guarantees.

 \subsection{Determining the next sampling instant}
 In this subsection, we discuss how Proposition~\ref{prop_hybrid} can be used at time $\tvar_j$ if $V(x(\tvar_j^+)) \leq c$ to maximize $\svar_{j+1} = \svar_j + \auxvar(s_j^+)$ such that \eqref{eq_idea_window2} holds.  We suppose subsequently that Assumption~\ref{asum_hybrid_lyap} holds for $V$, $\mathcal{X}\supseteq\mathcal{R}$ and $\mathcal{E}\supseteq \left\lbrace e = \hat{x}-x|\hat{x}\in \mathcal{R},x\in\mathcal{R}\right\rbrace$ for $n_\tpar \in\mathbb{N}$ different parameter sets $\gamma_i,\epsilon_i,L_i, i \in \left[1,\dots,n_\tpar\right]$.
 
 At least for one of the parameter sets, to which we assign the index $\itilde$, we require that $\epsilon_{\itilde} >0$ in order to obtain stability guarantees. For all further parameter sets with index $i \neq \itilde,$ $\epsilon_i$ may be negative, which will turn out to be very useful to maximize the time between sampling instants.
 
  Depending on the parameters, different values for $T_{\max} (\gamma,\lvar )$ can be obtained that allow to make statements about $V$ based on Proposition~\ref{prop_hybrid} for different time spans between sampling instants. Depending on the time span, either a decrease of $U$ and thus of $V$ can be guaranteed or, e.g, for negative $\epsilon$, a bound for the maximum increase of $V$ can be obtained.

Recall that we want to maximize $\auxvar(s_j^+)$ such that \eqref{eq_idea_window2} holds. For any of the parameter sets, we observe due to \eqref{eq_prop_hybrid1} and $V(x(t,j))\leq~U(\xi(t,j))~ \forall t,j$ that  \eqref{eq_idea_window2} holds if
\begin{align}
\nonumber &e^{\max\left\lbrace -\epsilon_i,2(L_i-\lvar_i ) \right\rbrace\auxvar(\tvar_j^+)} V(x(\tvar_j))\\
\leq& e^{-\epsilon_\refer\auxvar(\tvar_j^+)} \windowC \label{eq_dec_tr}
\end{align}
and $\auxvar(\tvar_j^+) < T_{\max}(\gamma_i,\Lambda_i)$ for some $i$. Thus, our goal is now to maximize $\auxvar(\tvar_j^+)$ such that there is at least one parameter set $i$ for which  $\eqref{eq_dec_tr}$ and $\auxvar(\tvar_j^+) < T_{\max}(\gamma_i,\Lambda_i)$ hold for some $\Lambda_i$. 

To be able to search efficiently for $\auxvar(s_j^+)$, we make two simplifications. Firstly, we replace the condition $\auxvar(\tvar_j^+) < T_{\max}(\gamma_i,\Lambda_i)$ by $\auxvar(\tvar_j^+) \leq \delta T_{\max}(\gamma_i,\Lambda_i)$ for some\footnote{Typically $\delta$ will be chosen close to $1$ (e.g. 0.999) in order to obtain a preferably large inter-sampling interval.\label{fn_1}} $\delta \in \left(0,1\right).$
Secondly, we fix\footnote{Note that chosing $\Lambda_i$ larger would not provide any advantage, whilst smaller choices could in some situations be advantageous. Since the expected advantage is typically minor, we omit it in the STC mechanism to reduce computational complexity. A line search could be included to exploit different values for $\Lambda_i$.}  $\Lambda_i = \max \left\lbrace L_i+\frac{\epsilon_i}{2}, 1-\delta \right\rbrace$.  Here the (typically small) positive value of $1-\delta$ avoids that $\Lambda_i \leq 0$. Note that Proposition~\ref{prop_hybrid} still applies despite the simplifications. 

The second simplification allows us to rewrite \eqref{eq_dec_tr} as
\begin{align}
 \left(-\epsilon_i +\epsilon_\refer\right)\auxvar(\tvar_j^+)
\leq \log(\frac{\genC}{V(x(\tvar_j))}).
\label{eq_ln}
\end{align}

Maximizing $\auxvar(\tvar_j^+)$ such that \eqref{eq_ln} and $\auxvar(\tvar_j^+) \leq \delta T_{\max}(\gamma_i,\Lambda_i)$ hold is straightforward.  It can be done in closed form by distinguishing  between the following four cases.
\begin{enumerate}[leftmargin=1.5 em]
	\item If $\genC \geq  V(x(\tvar_j))$ and $-\epsilon_i + \epsilon_\refer > 0$, then we obtain
	\begin{align*}
	\nonumber
	\auxvar(\tvar_j^+)  =& \min\left\lbrace \delta T_{\max}(\gamma_i,\Lambda_i)\vphantom{\frac{\log(\windowC)-\log(V(x(\tvar_j)))}{\max\left\lbrace -\epsilon_i,2(L_i-\lvar_i ) \right\rbrace + \epsilon_\refer}}\right.,\\
	&\left.\frac{\log(\genC)-\log(V(x(\tvar_j)))}{ -\epsilon_i + \epsilon_\refer} \right\rbrace.
	\end{align*}
	\item 	If $\genC \geq  V(x(\tvar_j))$ and $ -\epsilon_i + \epsilon_\refer \leq 0$, we obtain $\auxvar(\tvar_j^+) = \delta T_{\max}(\gamma_i,\Lambda_i)$.
	\item If $\genC < V(x(\tvar_j))$ and $ -\epsilon_i + \epsilon_\refer \geq 0$, then $\auxvar(\tvar_j^+)$ that solves \eqref{eq_dec_tr} would be negative. We are not interested in negative values for $\auxvar(\tvar_j^+)$ and therefore do not go into further detail here.
	\item If $\genC <  V(x(\tvar_j))$ and $ -\epsilon_i + \epsilon_\refer < 0,$ then \eqref{eq_dec_tr} would hold if $\auxvar(\tvar_j) \geq   \frac{\log(\genC)-\log(V(x(\tvar_j)))}{ -\epsilon_i + \epsilon_\refer} \coloneqq \bar t.$ Taking into account the additional condition  $\bar h_i \leq \delta T_{\max}(\gamma_i,\Lambda_i)$  
	we obtain in this case the solution $\auxvar(\tvar_j^+) = 	\delta T_{\max}(\gamma_i,\Lambda_i)$ if $\delta T_{\max}(\gamma_i,\Lambda_i) > \bar t$. If $\delta T_{\max}(\gamma_i,\Lambda_i) \leq \bar t$, then no solution exists.
\end{enumerate}
Note that Proposition~\ref{prop_hybrid} provides only a sufficient condition for \eqref{eq_prop_hybrid1}. Thus, if there is no positive value for $\auxvar(\tvar_j^+)$ for one specific parameter set for which \eqref{eq_ln} holds with $\auxvar(\tvar_j^+) \leq \delta T_{\max}(\gamma_i,\Lambda_i)$, then this does not mean that \eqref{eq_idea_window2} cannot be satisfied. Instead, different parameter sets can be tried in this case. To determine the next sampling instant $\svar_{j+1} = \auxvar(\tvar_{j}^+) + \svar_j$, the STC mechanism therefore probes all parameter sets for the largest possible value of $\auxvar(\tvar_{j}^+)$ for which it can be guaranteed that \eqref{eq_idea_window2} holds. This value is then used to determine the next sampling instant. 

\change{}{Whilst~\eqref{eq_idea_window2} would be sufficient to guarantee asymptotic stability of the set $\left\lbrace \left(x,e,\eta,\tau,\auxvar\right): x = 0, e= 0, \eta = 0 \right\rbrace$ if it would hold at all sampling instants, there may still be the situation, that we cannot guarantee that \eqref{eq_idea_window2} holds at a sampling instant. Thus, }we also need to handle the situation where we cannot guarantee for any parameter set that \eqref{eq_idea_window2} holds. In this case, we \change{will}{} exploit that $\epsilon_1 > 0$. If we set $\auxvar(s_j^+) = T_{\max}(\gamma_1,L_1+\frac{\epsilon_1}{2}) > 0$, then we can conclude from Proposition~\ref{prop_hybrid} that
\begin{equation}
\label{eq_dec_fb}
	V(x(s_{j+1})) \leq e^{-\epsilon_1 \delta T_{\max}\left(\gamma_1,L_1+\frac{\epsilon_1}{2}\right)} V(x(s_j)),
\end{equation}
i.e., a certain decrease of $V$ can be guaranteed in this case. \change{}{This decrease of $V$ can be used to ensure asymptotic stability of the set $\left\lbrace \left(x,e,\eta,\tau,\auxvar\right): x = 0, e= 0, \eta = 0 \right\rbrace$. } \change{This}{ Hence, the choice $\auxvar(s_j^+) = T_{\max}(\gamma_1,L_1+\frac{\epsilon_1}{2}) > 0$ } will be used as fall-back strategy if there is no parameter set for which \eqref{eq_ln} holds with $\auxvar(\tvar_j) > 0$, i.e., if we cannot guarantee for any parameter set that \eqref{eq_idea_window2} holds. 

The overall procedure to determine $\auxvar(s_j^+)$ is summarized in Algorithm~\ref{algo_trig_window}, which serves therefore as an implicit definition of  the function $\Gamma$ for the proposed STC mechanism. 
\begin{table}[tb]
	\begin{algorithm}[H]
		\caption[Caption for LOF]{Computation of $\auxvar = \Gamma(x(\tvar_j),\eta(\tvar_j))$ for the proposed STC mechanism for some $\delta \in \left(0,1\right)$ and $\epsilon_\refer > 0$.}
		\label{algo_trig_window}
		\begin{algorithmic}[1]
			\STATE $V \leftarrow V(x(\tvar_j))$, $C \leftarrow \genC$%
			\STATE $\bar h \leftarrow \delta T_{\max}(\gamma_\itilde,L_\itilde+\frac{\epsilon_\itilde}{2})$ \label{line_fallback}
			\FOR{\text{\bf each} $i \in \left[2,\dots,n_\tpar\right]$ } 
			\STATE $\Lambda_i \leftarrow \max \left\lbrace L_i + \frac{\epsilon_i}{2},(1-\delta) \right\rbrace$
			\IF{$C \geq V$} \label{line_for_start} %
			\IF{$-\epsilon_i+\epsilon_\refer > 0$}
			\STATE $\bar h_i \leftarrow \min\left\lbrace \delta T_{\max}(\gamma_i,\Lambda_i),				\frac{\log(C)-\log(V)}{ -\epsilon_i+ \epsilon_\refer} \right\rbrace$
			\ELSE
			\STATE $\bar h_i \leftarrow \delta T_{\max}(\gamma_i,\Lambda_i)$
			\ENDIF
			\ELSE
			\IF{$-\epsilon_i+\epsilon_\refer \geq 0$}
			\STATE $\bar h_i \leftarrow 0$
			\ELSE
			\STATE$ \bar t \leftarrow \frac{\log(C)-\log(V)}{ -\epsilon_i+ \epsilon_\refer}$
			\IF{$\bar t < \delta T_{\max}(\gamma_i,\Lambda_i)$}
			\STATE $\bar h_i \leftarrow \delta T_{\max}(\gamma_i,\Lambda_i)$
			\ELSE
			\STATE $\bar h_i \leftarrow 0$
			\ENDIF
			\ENDIF											
			\ENDIF\label{line_for_end}
			\IF{$\bar h_i > \bar h$}
			\STATE $\bar h \leftarrow \bar h_i$\label{line_h_update}
			\ENDIF
			\ENDFOR 			
			\STATE $\Gamma(x(\tvar_j),\eta(\tvar_j)) \leftarrow \bar{h}$ \label{line_gamma}
		\end{algorithmic}
	\end{algorithm}
	\vspace{-10mm}
\end{table}
The first step of the algorithm is to determine $V$ and $C$ based on $x(\tvar_j)$, $\eta(\tvar_j)$ and $c$. Then, a strictly positive minimum value for $\auxvar$ is determined using the fall-back strategy based on $\epsilon_{\itilde}, \gamma_{\itilde}$ and $L_{\itilde}$ and stored to $\bar{h}$.
After that, an iteration over all other parameter sets is started.  
For each parameter set, the algorithm computes a preferably large value $\bar h_i \leq \delta T_{\max}(\gamma_i,\Lambda_i)$ for which \eqref{eq_ln} holds for $\auxvar(\tvar_j^+)$  according to the previously discussed procedure. If such a value exists for a parameter set and exceeds the current value of $\bar h_i$, then the algorithm updates $\bar h_i$. If not, then the algorithm directly proceeds to the next parameter set, until all parameter sets have been investigated. 

Thus, Algorithm~\ref{algo_trig_window} ensures that either \eqref{eq_idea_window2} or \eqref{eq_dec_fb} hold. Moreover, it guarantees that $t_{\min} \geq \delta T_{\max}\left(\gamma_1,L_1+\frac{\epsilon_1}{2}\right)>0$. We will demonstrate in the next subsection how stability guarantees can be derived from these facts.

\subsection{Stability results}
In this subsection, we derive stability guarantees for the proposed dynamic STC mechanism based on the findings from the previous subsections. The steps carried out by the dynamic STC mechanism at a sampling instant are summarized in Algorithm~\ref{algo_stc}. We obtain the following stability result.

\begin{theo}
	\label{theo_stab}
	Assume there are $n_\tpar $ different parameter sets $\epsilon_i,\gamma_i,L_i$, $i\in\left[1,\dots,n_\tpar\right]$, for which Assumption~\ref{asum_hybrid_lyap} holds for the same function $V(x)$, $\mathcal{X} \supseteq \mathcal{R}$ and $\mathcal{E}\supseteq \left\lbrace e = \hat{x}-x|\hat{x}\in \mathcal{R},x\in\mathcal{R}\right\rbrace$. Let $\epsilon_1 > 0$. Consider \eqref{eq_sys_hyb} with $S(\eta,x)$  and $\Gamma(x,\eta)$ defined according to \eqref{eq_S_window}  and by Algorithm~\ref{algo_trig_window} with some $\delta\in\left(0,1\right)$ and $\epsilon_\refer > 0$. 
	Then $\svar_{j+1} -\svar_j \geq t_{\min} \coloneqq \delta T_{\max}(\gamma_1,L_1+\frac{\epsilon_1}{2})$ for all $j\in\mathbb{N}_0$ and the set $\left\lbrace \left(x,e,\eta,\tau,\auxvar\right): x = 0, e= 0, \eta = 0 \right\rbrace$ is asymptotically stable with region of attraction $\mathcal{R}$.
\end{theo}
\begin{proof}
	See Appendix~\ref{proof_theo}.
\end{proof}
If Assumption~\ref{asum_hybrid_lyap} holds globally, i.e., with $\mathcal{X} = \mathbb{R}^{n_x}$ and $\mathcal{E} = \mathbb{R}^{n_e}$, then global stability guarantees can be obtained for the dynamic STC mechanism, as stated in the following corollary of Theorem~\ref{theo_stab}.
\begin{coro}
	Let the conditions of Theorem~\ref{theo_stab} hold for $\mathcal{R} = \mathbb{R}^{n_x}$. Then $\svar_{j+1} -\svar_j \geq t_{\min} \coloneqq \delta T_{\max}(\gamma_1,L_1+\frac{\epsilon_1}{2})$ for all $j\in\mathbb{N}_0$ and the set $\left\lbrace \left(x,e,\eta,\tau,\auxvar\right): x = 0, e= 0, \eta = 0 \right\rbrace$ is uniformly globally asymptotically stable.
\end{coro}
\begin{rema}
	\label{rk_ic_window}
	A suitable choice for $\eta(\tvar_0)$ is, e.g., $\eta_k(\tvar_0) = V(\tvar_0)$ for all $k\in\left[1,\dots,m-1\right]$.
\end{rema}
\begin{rema}
	The computational effort for determining the next sampling instant depends linearly on $n_\tpar$. Typically, the proposed mechanism becomes less conservative as the number of different parameter sets increases. Thus, we can typically trade off the computational effort and the achievable length of the sampling intervals.
\end{rema}

	\begin{rema}
		The proposed dynamic STC mechanism can be modified in order to relax Assumption~\ref{asum_hybrid_lyap} by replacing \eqref{eq_v_desc_hybrid} with a more general ISS condition with a class $\mathcal{K}$ function $\varrho(\norm{x})$ instead of $\epsilon V(x)$. Then, the decrease condition in Proposition~\ref{prop_hybrid} can be replaced by 
		$\frac{d}{dt} U(\xi(t,j)) \leq \max \left\lbrace\varrho\left(\norm{x(t,j)}\right),2(L-\Lambda)U(\xi(t,j)) \right\rbrace,$
		which can still be used to construct a bound on the future evolution of $V(x(t,j))$. Algorithm~\ref{algo_trig_window} can be modified to incorporate this bound, leading to similar stability guarantees.
	\end{rema}

\begin{table}[tb]
	\begin{algorithm}[H]
		\caption[Caption for LOF]{Dynamic STC mechanism at sampling instant $t_j$}
		\label{algo_stc}
		\begin{algorithmic}[1]
			\STATE Measure $x_p(\tvar_j)$
			\STATE Set $\hat{x}_p(\tvar_j^+) = x_p(\tvar_j)$
			\STATE Set $\hat{u}(\tvar_j^+) = g_c(x_c(\tvar_j),x_p(\tvar_j)$
			\STATE Set $\auxvar(\tvar_j^+) = \Gamma(x(\tvar_j),\eta(\tvar_j))$ with $\Gamma$ defined by Algorithm~\ref{algo_trig_window}
			\STATE Set $\eta(\tvar_j^+) = S(\eta(\tvar_j),x(\tvar_j))$
			\STATE Wait until $\svar_{j+1}$
		\end{algorithmic}
	\end{algorithm}
	\vspace{-5mm}
\end{table}
\section{Example}
\label{sec_ex}

In this section, we illustrate the proposed dynamic STC mechanism with a numerical example from the literature. We demonstrate that the proposed mechanism can significantly reduce the required amount of sampling instants in comparison to existing static mechanisms from the literature.

As an example system, we use the forced Van der Pol oscillator that was also considered in \cite{delimpaltadakis2020isochronous}. The dynamics of the plant are given by
\begin{equation*}
	\begin{split}
		\dot{x}_{p1} =& x_{p2}\\
		\dot{x}_{p2} =& (1-x_{p1}^2)x_{p2}-x_{p1}+\hat{u}.
	\end{split}
\end{equation*}

The plant is controlled by the static state-feedback controller
	$u = -x_{p2}-(1-x_{p1}^2)x_{p2},$
leading for  $a_1(x,e) = 2x_{p1}e_2+e_1e_2$ and $a_2(x,e) = x_{p1}^2+2x_{p1}e_1+e_1^2$ to $f(x,e) = \begin{bmatrix}
0 & 1\\
-1 & -1
\end{bmatrix} x + \begin{bmatrix}
0 & 0\\
a_1(x,e) & -2+a_2(x,e)
\end{bmatrix} e$,
and $g(x,e) = -f(x,e)$. We consider the function $V(x) = x^\top P x$ with $P = \begin{bmatrix}
4.68 & 1.10\\
1.10 & 3.56
\end{bmatrix}$. We are interested in guaranteeing asymptotic stability for the region of attraction $\mathcal{R} \coloneqq \left\lbrace x | V(x) \leq 10 \right\rbrace$, i.e., we chose $c = 10$.

To this end, we need to determine parameters for which Assumption~\ref{asum_hybrid_lyap} holds for $\mathcal{X} \supseteq \mathcal{R}$ and $\mathcal{E} \supseteq \left\lbrace e = \hat{x}-x|\hat{x}\in \mathcal{R},x\in\mathcal{R}\right\rbrace$. In particular, we verify the assumption for the sets $\mathcal{X} \coloneqq \left\lbrace x | \norm{x} \leq \bar a \right\rbrace \supseteq \mathcal{R}$ and $\mathcal{E} \coloneqq \left\lbrace e | \norm{e} \leq 2 \bar a \right\rbrace \supseteq  \left\lbrace e = \hat{x}-x|\hat{x}\in \mathcal{R},x\in\mathcal{R}\right\rbrace$ for $\bar a = \sqrt{\frac{c}{\lambda_{\min}(P)}}$. Note that for $x \in \mathcal{X} $ and $e \in \mathcal{E}$, it holds that  $a_1(x,e) \in \left[-4\bar a,4\bar a \right]$ and $a_2(x,e) \in \left[-4\bar a,9\bar a\right]$. We can thus use the approach from \cite[Section~IV]{hertneck20simple} with the modification that we minimize $\gamma$ subject to constraint\footnote{To account for the exponential ISS condition \eqref{eq_v_desc_hybrid}, we have to replace $\Lambda_{11}$ by $\Lambda_{11} = 	A^\top P + PA +\epsilon P+M_w^2C^\top C$.} \cite[Equation~(21)]{hertneck20simple} for all combinations of the bounds on $a_1(x,e)$ and $a_2(x,e)$ simultaneously to determine the values of $\gamma$ and $L$ for a given $\epsilon$. 

In total, we have computed $n_\tpar = 21$ different parameter sets for different values of $\epsilon_i, i \in \left[2,\dots,n_\tpar\right]$ in the range of $-40 \leq \epsilon_i \leq 0.01$. Here $\epsilon_1 = 0.01$ is also used for the fall-back strategy, for which the corresponding parameter set leads to $t_{\min} =  \SI{0.022}{\second} $.

\begin{figure}
	\centering
	\resizebox{\linewidth}{!}{
		\input{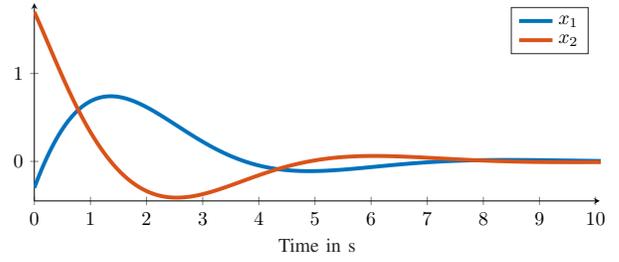}
	}
	\vspace{-3mm}
	\caption{State trajectory for the dynamic STC mechanism for $x_{p}(0) = [-0.3,1.7]^\top$.
	}
	\label{fig_states}
	\vspace{-3mm}
\end{figure}
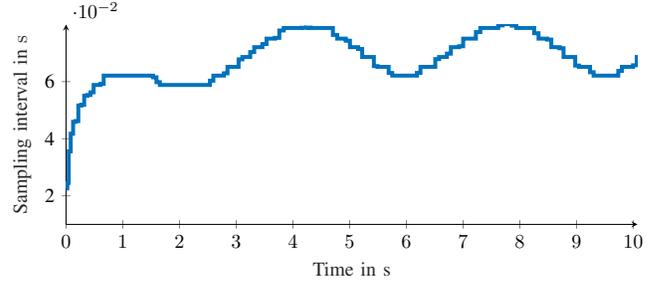
\begin{figure}
	\centering
	\resizebox{\linewidth}{!}{
		\definecolor{mycolor1}{rgb}{0.00000,0.44700,0.74100}%
\begin{tikzpicture}

\begin{axis}[%
width=4in,
height=1.4in,
at={(0in,0in)},
scale only axis,
axis lines=left,
xmin=0,
xmax=10.07565703509329,
xlabel style={font=\color{white!15!black}},
xlabel={Time in s},
ylabel style={font=\color{white!15!black}},
ylabel={Sampling interval in \SI{}{\second}},
ymin=0.01,
ymax=0.08,
ytick={0.02, 0.04, 0.06},
yticklabels = {2,4,6},
axis background/.style={fill=white},
axis x line*=bottom,
axis y line*=left,
legend style={legend cell align=left, align=left, draw=white!15!black}
]
\addplot[const plot, color=mycolor1, line width=2.0pt] table[row sep=crcr] {%
0	0.0226\\
0.0226	0.024532561866256\\
0.047132561866256	0.0355819558273184\\
0.0827145176935744	0.0417484072420151\\
0.12446292493559	0.0459267192998759\\
0.170389644235465	0.0461841179485069\\
0.216573762183972	0.0516447797263893\\
0.268218541910362	0.0518702100958133\\
0.320088752006175	0.05505287146464\\
0.375141623470815	0.0553136252309825\\
0.430455248701798	0.0562469206009806\\
0.486702169302778	0.0588314179981555\\
0.545533587300934	0.0588314179981555\\
0.604365005299089	0.0592035547639157\\
0.663568560063005	0.0621035195525551\\
0.72567207961556	0.0621035195525551\\
0.787775599168115	0.0621035195525551\\
0.84987911872067	0.0621035195525551\\
0.911982638273225	0.0621035195525551\\
0.97408615782578	0.0621035195525551\\
1.03618967737834	0.0621035195525551\\
1.09829319693089	0.0621035195525551\\
1.16039671648345	0.0621035195525551\\
1.222500236036	0.0621035195525551\\
1.28460375558856	0.0621035195525551\\
1.34670727514111	0.0621035195525551\\
1.40881079469367	0.0621035195525551\\
1.47091431424622	0.0621035195525551\\
1.53301783379878	0.0615223863686578\\
1.59454022016744	0.0598130683276127\\
1.65435328849505	0.0588314179981555\\
1.7131847064932	0.0588314179981555\\
1.77201612449136	0.0588314179981555\\
1.83084754248951	0.0588314179981555\\
1.88967896048767	0.0588314179981555\\
1.94851037848583	0.0588314179981555\\
2.00734179648398	0.0588314179981555\\
2.06617321448214	0.0588314179981555\\
2.12500463248029	0.0588314179981555\\
2.18383605047845	0.0588314179981555\\
2.2426674684766	0.0588314179981555\\
2.30149888647476	0.0588314179981555\\
2.36033030447291	0.0588314179981555\\
2.41916172247107	0.0588314179981555\\
2.47799314046923	0.0588314179981555\\
2.53682455846738	0.060323896161723\\
2.5971484546291	0.0621035195525551\\
2.65925197418166	0.0621035195525551\\
2.72135549373422	0.0621035195525551\\
2.78345901328677	0.0629196986756708\\
2.84637871196244	0.0651785882631059\\
2.91155730022555	0.0651785882631059\\
2.97673588848865	0.0651785882631059\\
3.04191447675176	0.0678006471997561\\
3.10971512395151	0.0686269817202163\\
3.17834210567173	0.0686269817202163\\
3.24696908739195	0.0703514697626841\\
3.31732055715463	0.0720679844979762\\
3.38938854165261	0.0720679844979762\\
3.46145652615058	0.0731682006910079\\
3.53462472684159	0.0750570057966924\\
3.60968173263828	0.0750570057966924\\
3.68473873843498	0.0753214159793271\\
3.7600601544143	0.077906860130253\\
3.83796701454456	0.0787848706612706\\
3.91675188520583	0.0787848706612706\\
3.9955367558671	0.0787848706612706\\
4.07432162652837	0.0787848706612706\\
4.15310649718964	0.0788554674014813\\
4.23196196459112	0.0789157294991228\\
4.31087769409024	0.0787848706612706\\
4.38966256475151	0.0787848706612706\\
4.46844743541278	0.0787848706612706\\
4.54723230607406	0.0787848706612706\\
4.62601717673533	0.0787848706612706\\
4.7048020473966	0.0763495300764399\\
4.78115157747304	0.0750570057966924\\
4.85620858326973	0.0750570057966924\\
4.93126558906642	0.0743293871914699\\
5.00559497625789	0.0720679844979762\\
5.07766296075587	0.0720679844979762\\
5.14973094525384	0.0714216388517653\\
5.22115258410561	0.0686269817202163\\
5.28977956582582	0.0686269817202163\\
5.35840654754604	0.0686269817202163\\
5.42703352926626	0.0654353084399445\\
5.4924688377062	0.0651785882631059\\
5.55764742596931	0.0651785882631059\\
5.62282601423241	0.0651785882631059\\
5.68800460249552	0.0630232975689889\\
5.75102790006451	0.0621035195525551\\
5.81313141961706	0.0621035195525551\\
5.87523493916962	0.0621035195525551\\
5.93733845872218	0.0621035195525551\\
5.99944197827473	0.0621035195525551\\
6.06154549782729	0.0621035195525551\\
6.12364901737984	0.0621035195525551\\
6.1857525369324	0.0630439163684112\\
6.24879645330081	0.0651785882631059\\
6.31397504156391	0.0651785882631059\\
6.37915362982702	0.0651785882631059\\
6.44433221809013	0.0651785882631059\\
6.50951080635323	0.067889035559049\\
6.57739984191228	0.0686269817202163\\
6.6460268236325	0.0686269817202163\\
6.71465380535271	0.0698119700688354\\
6.78446577542155	0.0720679844979762\\
6.85653375991953	0.0720679844979762\\
6.9286017444175	0.0723946217602832\\
7.00099636617778	0.0750570057966924\\
7.07605337197448	0.0750570057966924\\
7.15111037777117	0.0750570057966924\\
7.22616738356786	0.0773472915098005\\
7.30351467507766	0.0787848706612706\\
7.38229954573893	0.0787848706612706\\
7.4610844164002	0.0787848706612706\\
7.53986928706147	0.0787848706612706\\
7.61865415772275	0.0796821500525565\\
7.6983363077753	0.0801588330170585\\
7.77849514079236	0.0801372606014653\\
7.85863240139383	0.0795943180663088\\
7.93822671946013	0.0787848706612706\\
8.0170115901214	0.0787848706612706\\
8.09579646078267	0.0787848706612706\\
8.17458133144394	0.0787848706612706\\
8.25336620210521	0.0769422860368376\\
8.33030848814205	0.0750570057966924\\
8.40536549393874	0.0750570057966924\\
8.48042249973544	0.0747204631305691\\
8.55514296286601	0.0720679844979762\\
8.62721094736398	0.0720679844979762\\
8.69927893186196	0.0716919976297806\\
8.77097092949174	0.0686269817202163\\
8.83959791121196	0.0686269817202163\\
8.90822489293217	0.0686269817202163\\
8.97685187465239	0.0656204520116038\\
9.04247232666399	0.0651785882631059\\
9.1076509149271	0.0651785882631059\\
9.1728295031902	0.0651785882631059\\
9.23800809145331	0.0632946949049575\\
9.30130278635827	0.0621035195525551\\
9.36340630591082	0.0621035195525551\\
9.42550982546338	0.0621035195525551\\
9.48761334501593	0.0621035195525551\\
9.54971686456848	0.0621035195525551\\
9.61182038412104	0.0621035195525551\\
9.67392390367359	0.0621035195525551\\
9.73602742322615	0.0637581498235002\\
9.79978557304965	0.0651785882631059\\
9.86496416131275	0.0651785882631059\\
9.93014274957586	0.0651785882631059\\
9.99532133783896	0.0657283401345322\\
10.0610496779735	0.0686269817202163\\
10.1296766596937	0.0686269817202163\\
10.1983036414139	0.0686269817202163\\
10.2669306231341	0.0707340197923155\\
10.3376646429265	0.0720679844979762\\
10.4097326274244	0.0720679844979762\\
10.4818006119224	0.07328134605536\\
10.5550819579778	0.0750570057966924\\
10.6301389637745	0.0750570057966924\\
10.7051959695712	0.0754236554645068\\
10.7806196250357	0.0780848334845874\\
10.8587044585203	0.0787848706612706\\
10.9374893291815	0.0787848706612706\\
11.0162741998428	0.0787848706612706\\
11.0950590705041	0.0790824564101074\\
11.1741415269142	0.0798997897412459\\
11.2540413166554	0.0802489380685233\\
11.3342902547239	0.0800915208784057\\
11.4143817756023	0.0794136595783456\\
11.4937954351807	0.0787848706612706\\
11.572580305842	0.0787848706612706\\
11.6513651765032	0.0787848706612706\\
11.7301500471645	0.0787848706612706\\
11.8089349178258	0.0762641355803801\\
11.8851990534062	0.0750570057966924\\
11.9602560592028	0.0750570057966924\\
12.0353130649995	0.073880164577304\\
12.1091932295768	0.0720679844979762\\
12.1812612140748	0.0720679844979762\\
12.2533291985728	0.0707855876537812\\
12.3241147862266	0.0686269817202163\\
12.3927417679468	0.0686269817202163\\
12.461368749667	0.068126586883372\\
12.5294953365504	0.0651785882631059\\
12.5946739248135	0.0651785882631059\\
12.6598525130766	0.0651785882631059\\
12.7250311013397	0.0648981977410241\\
12.7899292990807	0.0628319989736414\\
12.8527612980544	0.0621035195525551\\
12.9148648176069	0.0621035195525551\\
12.9769683371595	0.0621035195525551\\
13.039071856712	0.0621035195525551\\
13.1011753762646	0.0621035195525551\\
13.1632788958171	0.0621035195525551\\
13.2253824153697	0.0624480996293343\\
13.287830514999	0.064343751093843\\
13.3521742660929	0.0651785882631059\\
13.417352854356	0.0651785882631059\\
13.4825314426191	0.0651785882631059\\
13.5477100308822	0.0665443136139\\
13.6142543444961	0.0686269817202163\\
13.6828813262163	0.0686269817202163\\
13.7515083079365	0.0686269817202163\\
13.8201352896567	0.0716201710227067\\
13.8917554606794	0.0720679844979762\\
13.9638234451774	0.0720679844979762\\
14.0358914296754	0.0741363253994818\\
14.1100277550749	0.0750570057966924\\
14.1850847608716	0.0750570057966924\\
14.2601417666683	0.076157310080477\\
14.3362990767487	0.0787848706612706\\
14.41508394741	0.0787848706612706\\
14.4938688180713	0.0787848706612706\\
14.5726536887325	0.0787848706612706\\
14.6514385593938	0.0793721027160132\\
14.7308106621098	0.080067231045706\\
14.8108778931555	0.0802804900399157\\
14.8911583831954	0.0799807378754468\\
14.9711391210709	0.0791619372119329\\
15.0503010582828	0.0787848706612706\\
};
\end{axis}

\end{tikzpicture}%
}
	\vspace{-4mm}
	\caption{Sampling intervals for the dynamic STC mechanism for the trajectory starting at $x_{p}(0) = [-0.3,1.7]^\top$.
	}
	\label{fig_inters}
\end{figure}
\begin{figure}
	\centering
	\resizebox{\linewidth}{!}{
		\input{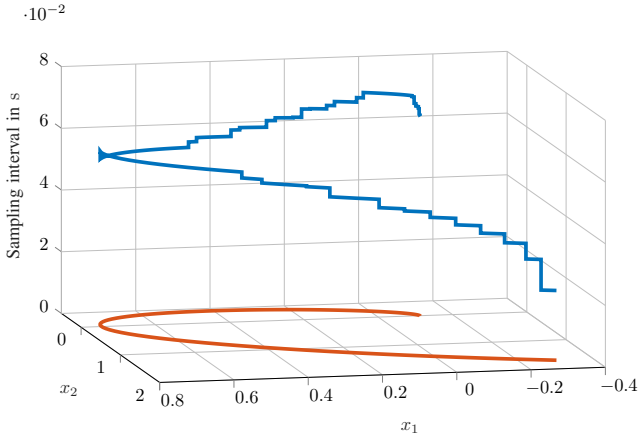}}
	\vspace{-4mm}
	\caption{Evolution of the sampling intervals along the state trajectory starting at $x_p(0) = [-0.3,1.7]^\top$ \change{}{for the first \SI{5}{\second}}.} 
	\label{fig_traj}
	\vspace{-3mm}
\end{figure}

In Figure~\ref{fig_states}, a state trajectory for a simulation with $\epsilon_\refer = 0.01, m = 30$ and $\eta(\tvar_0)$ chosen according to Remark~\ref{rk_ic_window} starting at $x_p(0) = [-0.3,1.7]^\top$ is depicted. Figure~\ref{fig_inters} shows the time evolution of the resulting sampling intervals. Figure~\ref{fig_traj} displays the evolution of the sampling intervals along the considered trajectory. It can be seen that the sampling interval rapidly increases at the beginning. After roughly \SI{1}{\second} it starts to oscillate slowly between \SI{0.05}{\second} and \SI{0.07}{\second}. When simulated for a longer amount  of time, this oscillating behavior continues, which is not surprising as the considered system exhibits oscillating behavior as well.

The considered example allows also a comparison of the proposed dynamic STC mechanism to the static STC mechanism from \cite{delimpaltadakis2020isochronous}. For the considered trajectory, the dynamic STC mechanism produces significantly larger sampling intervals at roughly the same control performance. 
Whilst in \cite{delimpaltadakis2020isochronous}, a total number of 1448 sampling instants is reported during the first \SI{5}{\second} of the simulation, the proposed dynamic STC requires only 89 sampling instants, which constitutes a reduction by a factor 16. Nevertheless, the state trajectories for both mechanisms are similar\change{}{, which can, e.g., be seen when comparing Fig.~\ref{fig_traj} to \cite[Fig. 9]{delimpaltadakis2020isochronous}}.
This demonstrates that dynamic STC can lead to  significantly larger sampling intervals than static STC approaches. \change{}{Note that it might also be possible to emulate a dynamic ETC condition in the setup of \cite{delimpaltadakis2020isochronous,delimpaltadakis2020region}, leading to less conservative results.}
\section{Summary and outlook}
In this paper, we have presented a dynamic STC mechanism that takes into account the past system behavior using a dynamic variable that stores the past evolution of the Lyapunov function. Using a hybrid Lyapunov function approach, sampling instants are chosen such that an average decrease of the Lyapunov function is ensured leading to guarantees for asymptotic stability. In a numerical example, it was demonstrated that the proposed dynamic STC mechanism can reduce the required amount of samples significantly in comparison to existing static STC approaches.  

Future work will address different strategies to update the dynamic variable in order to adjust the behavior of the STC mechanism, e.g., for setpoint changes. Moreover, we are interested in obtaining stability and performance guarantees also for nonlinear systems subject to perturbations.
\label{sec_sum}

\bibliography{../../../../Literatur/literature}
\appendix
\subsection{Proof of Proposition~\ref{prop_hybrid}}
\label{app_hyb}
	This proof follows the same lines as the proof of \cite[Proposition~12]{hertneck20stability} but is adapted to the local setup considered in this paper.
	Recall from \cite{carnevale2007lyapunov} that $\phi(\tau) \in \left[\lambda, \lambda^{-1}\right]$ for all $\tau \in \left[0, \tilde{T}_{\max} \right]$, where $\tilde{T}_{\max} = \tilde{T}_{\max}(\lambda,\gamma,\lvar )$ with		
		\begin{equation*}
			\footnotesize
		\tilde{T}_{\max}(\lambda,\gamma,\lvar ) \defeq \begin{cases}\vspace{1mm}
		\frac{1}{\lvar r} \mathrm{arctan}\left(\frac{r(1-\lambda)}{2 \frac{\lambda}{1+\lambda} \left(\frac{\gamma}{\lvar }-1\right)+1+\lambda}\right) & \gamma > \lvar \\ \vspace{1mm}
		\frac{1}{\lvar } \frac{1-\lambda}{1+\lambda} & \gamma = \lvar \\
		\frac{1}{\lvar r} \mathrm{arctanh}\left(\frac{r(1-\lambda)}{2 \frac{\lambda}{1+\lambda} \left(\frac{\gamma}{\lvar }-1\right)+1+\lambda}\right) &\gamma < \lvar 
		\end{cases}	
		\normalsize
		\end{equation*}

	and  $r$ defined as in \eqref{eq_def_r}.

	Now, we consider \eqref{eq_sys_hyb}. For each $\auxvar(s_j^+) < T_{\max} (\gamma,\lvar )$, there is a $\lambda \in \left(0,1\right)$ such that $\auxvar(s_j^+)  = \tilde{T}_{\max}(\lambda,\gamma,\lvar )$.
	Similar as in \cite{nesic2009explicit}, we can conclude that
	\begin{equation}
	\label{eq_U_lyap}
		\underline{\alpha}_U\left(\norm{\left[x^\top,e^\top\right]^\top}\right) \leq U(\xi) \leq \overline{\alpha}_U\left(\norm{\left[x^\top,e^\top\right]^\top}\right)
	\end{equation}
	 for some $\underline{\alpha}_U, \overline{\alpha}_U \in \mathcal{K}_\infty$. We observe due to Assumption~\ref{asum_hybrid_lyap} that for $\tau \in \left[0,\auxvar(s_j^+)\right]$, almost all $x\in\mathcal{X}$ and almost all $e \in \mathcal{E}$
	\begin{equation*}
	\begin{split}
	\left\langle \nabla U(\xi), F(\xi) \right\rangle
	\leq&-\epsilon V(x) -H^2(x,e) + \gamma^2 W^2(e) \\
	&+2\gamma\phi(\tau)W(e)(L W(e) + H(x,e))\\
	& - \gamma W^2(e) (2\lvar \phi(\tau)+\gamma(\phi^2(t)+1))\\
	\leq&-\epsilon V(x) -(H(x,e) -\gamma\phi(\tau)W(e))^2\\ 
	&+ 2\gamma\phi(\tau)W^2(e)(L-\lvar )\\
	\leq& -\epsilon V(x) + 2\gamma\phi(\tau)W^2(e)(L-\lvar ).
		\end{split}
	\end{equation*} 
 	Since $x(s_j^+) \in \mathcal{X}$ and $e(s_j^+) = 0 \in \mathcal{E}$, we obtain for $t$ sufficiently close to $\svar_j$
 	\begin{equation}
 			\label{eq_dec_u}
 			\begin{split}
	 			&\frac{d}{dt} U(\xi(\svar,j+1))\\ \leq& -\epsilon V(x(\svar,j+1)) + 2\gamma\phi(\tau(\svar,j))W^2(e(\svar,j))(L-\lvar )
 			\end{split}	
 	\end{equation} 
 	and thus due to the comparison Lemma (cf. \cite[p. 102]{khalil2002nonlinear})
 	\begin{align}
 	&U(\xi(\svar,j+1)) \nonumber\\ 
 	\leq& \exp\left( \max \left\lbrace -\epsilon, 2(L - \lvar ) \right\rbrace (t-\svar_j)\right)U(\xi(\tvar_j^+))\nonumber \\
 	\leq& \exp\left( \max \left\lbrace -\epsilon, 2(L - \lvar ) \right\rbrace (t-\svar_j)\right)V(x(\tvar_j^+)). \nonumber
 	\end{align}
 	Hence, $x(\svar,j+1) \in \mathcal{X}$ for $t$ sufficiently close to $\svar_j$. 	Note that $x(\tvar_j^+) \in \mathcal{X}$ and that $e(\svar,j) = x(\tvar_j^+) - x(\svar,j+1)$. It therefore also holds that $e(\svar,j+1) \in \mathcal{E}$ for $t$ sufficiently close to $\svar_j$. We can now use this argumentation iteratively \change{}{together with the fact that $V(x(t) \leq U(\xi(t)))$ } to observe that \eqref{eq_prop_hybrid1} holds.
	\hfill\hfill \qed
\subsection{Proof of Theorem~\ref{theo_stab}}
\label{proof_theo}
		We will now first show that a bound on the evolution of $V(x(\tvar_j))$ follows from the choice of $\auxvar(s_j^+)$ for each $j \in \mathbb{N}$ according to Algorithm~\ref{algo_trig_window}. Then we will use this bound to construct explicitly the function $\beta$.
		
		Recall that 
		 the flow and jump sets of the hybrid system \eqref{eq_sys_hyb} are designed such that the jumps of the system occur exactly at sampling instants that are described by $\tvar_j \coloneqq  (\svar_j,j)$ and by $\tvar_j^+ = (\svar_j,j+1)$, where $\tvar_0 = (0,0)$ due to the sampling at $t_0 = 0$ (modeled in \eqref{eq_sys_hyb} by the initial condition $\tau(0,0) = \auxvar(0,0)$). Therefore, it holds that  $e(\tvar_j^+) = 0, \eta(\tvar_j^+)=S(\eta(\tvar_j),x(\tvar_j)), \tau(\tvar_j^+) = 0$ and  $\auxvar(\tvar_j^+)=\Gamma(x(\tvar_j),\eta(\tvar_j))$, where $\Gamma(x(\tvar_j,\eta(\tvar_j)))$ is defined by Algorithm~\ref{algo_trig_window}. 
		 
		 Obviously $\bar h \geq \delta T_{\max}(\gamma_1,L_1+\frac{\epsilon_1}{2}) = t_{\min}$ in Algorithm~\ref{algo_trig_window}, and thus it follows that $t_{j+1} - t_j =  \Gamma(x(\tvar_j),\eta(\tvar_j)) \geq t_{\min}$.
		 
		 Now suppose that $x(\tvar_j^+) \in \mathcal{R}$.  If Algorithm~\ref{algo_trig_window} outputs $\Gamma(x(\tvar_j),\eta(\tvar_j)) =  \delta T_{\max}(\gamma_1,L_1+\frac{\epsilon_1}{2})$, then $c \geq \max\left\lbrace 1, e^{\max\left\lbrace -\epsilon_1,2(L_1-\lvar_1 ) \right\rbrace(\auxvar(s_j^+))} \right\rbrace V(x(\tvar_j^+)$. Thus it follows in this case directly from Proposition~\ref{prop_hybrid} that 
		 \begin{equation}
		 \label{eq_theo_case1}
		 V(x(\tvar_{j+1})) \leq e^{-\epsilon_{\itilde} \auxvar(\tvar_j^+)} V(x(\tvar_j)).
		 \end{equation}

		 If Algorithm~\ref{algo_trig_window} outputs $\bar h = \bar h_i > \delta T_{\max}(\gamma_1,L_1+\frac{\epsilon_1}{2})$ for some $i \in \left[2,\dots,n_\tpar\right]$, then we know (cf. the discussion preceding Algorithm~\ref{algo_trig_window}) that \eqref{eq_dec_tr} holds for $\auxvar(s_j^+)$. Note that this implies since $\genC \leq c$ that $c \geq \max\left\lbrace 1, e^{\max\left\lbrace -\epsilon_i,2(L_i-\lvar_i ) \right\rbrace(\auxvar(s_j^+))} \right\rbrace V(x(\tvar_j^+)$. 
		 
		 We can thus conclude in this case from Proposition~\ref{prop_hybrid} that 
		 \begin{align}
		 \nonumber &V(x(\tvar_{j+1}))\\
		  \leq&\min \left\lbrace e^{-\epsilon_\refer\auxvar(\tvar_j^+)}\frac{1}{m} \left(V(x(\tvar_j))+\sum_{k=1}^{m-1} \eta_k(\tvar_j)\right),c\right\rbrace. 		 \label{eq_theo_case2}
		 \end{align}

		 Since either \eqref{eq_theo_case1} or \eqref{eq_theo_case2} hold, we can thus infer that 
		 \begin{equation}
		 \label{eq_V_dec_theo}
		 	\resizebox{\linewidth}{!}{$V(x(\tvar_{j+1})) \leq e^{-\tilde \epsilon t_{\min}} \max \left\lbrace V(x(\tvar_{j})),\eta_1(\tvar_j),\dots,\eta_{m-1}(\tvar_j) \right\rbrace$}
		 \end{equation}
		 where $\tilde{\epsilon} \coloneqq \min\left\lbrace \epsilon_1,\epsilon_\refer\right\rbrace > 0$ and that, if $x(\tvar_j^+) \in \mathcal{R}$, then $x(\tvar_{j+1}) = x(\tvar_{j+1}^+) \in \mathcal{R}$.
		Since $x(\tvar_0) = x_0 \in \mathcal{R}$, we obtain thus by induction for all $j\in \mathbb{N}_0$ that $x(\tvar_j) \in \mathcal{R}$ and that 
		 		 \begin{equation}
		 		 \label{eq_V_bound_theo}
		 		 		V(x(\tvar_j)) \leq \max \left\lbrace V(x(\tvar_0)),\eta_1(\tvar_0),\dots,\eta_{m-1}(\tvar_0) \right\rbrace.
		 		 \end{equation}  
		   Next, we show that $V(x(\tvar_j))$ is bounded by an exponentially decaying bound that will help us to construct $\beta$. The update of $\eta(\tvar_j)$ according to \eqref{eq_S_window} implies for $j>m$ and each $k\in\left[1,\dots,m-1\right]$ that $\eta_k(\tvar_j) = V(x(\tvar_{j-m+k}))$. With \eqref{eq_V_dec_theo} we can conclude for $k \in \left[1,\dots,m-1\right]$   
		 \begin{align}
	\nonumber
		 	&V(x(\tvar_{lm+k})) \\ 
		 	\leq& e^{-\tilde \epsilon t_{\min}} \max \left\lbrace V(x(s_{(l-1)m+k})),\dots, V(x(\tvar_{lm+k-1})) \right\rbrace \label{eq_V_desc_theo}
		 \end{align}
		 for all $l \in \mathbb{N}$. Using \eqref{eq_V_desc_theo} recursively, we thus obtain  
		 with~\eqref{eq_V_bound_theo}
		 \begin{align}
		 \nonumber	&V(x(\tvar_{lm+k}))\\
		 \nonumber \leq&  e^{-\tilde \epsilon t_{\min} l} \max \left\lbrace V(x(\tvar_{1})),\dots, V(x(\tvar_{m})) \right\rbrace\\
		 \nonumber				   \leq& e^{-\tilde \epsilon t_{\min} l} \max \left\lbrace V(x(\tvar_0)),\eta_1(\tvar_0),\dots,\eta(m-1)(\tvar_0) \right\rbrace\\
						   \leq& e^{-\tilde \epsilon t_{\min} l} \max \left\lbrace \bar{\alpha}_V(\norm{x(s_0)}),\norm{\eta(s_0)} \right\rbrace. \label{eq_Vk_bound_theo}
		 \end{align} 
		 Now we chose $k$ and $l$ such that $j = lm+k$, which requires $l \geq \frac{j-m}{m}$. Thus, we obtain from \eqref{eq_Vk_bound_theo}
		 \begin{equation}
		 \label{eq_V_bound_new}
		 	V(x(\tvar_{j}))  
		 	\leq e^{\frac{ -\tilde\epsilon t_{\min}}{m} \left( j-m\right)} \max \left\lbrace \bar{\alpha}_V(\norm{x(s_0)}),\norm{\eta(s_0)} \right\rbrace.
		 \end{equation} 
		 Note that
		 \begin{equation*}
		 	\svar_{j+1} - \svar_j = \auxvar(\tvar_j^+) \leq t_{\max} \coloneqq \underset{i \in \left[1,\dots,n_\tpar\right]}{\max} \delta T_{\max}(\gamma_i,L_i+\frac{\epsilon_i}{2})
		 \end{equation*} holds for any $j \in \mathbb{N}_0$, which implies $\frac{j}{2} \geq \frac{\svar_j}{2t_{\max}}$. Using this in \eqref{eq_V_bound_new} we obtain
		 \newlength\mylen
		 \settoheight\mylen{$\max \left\lbrace \underline{\alpha}_V^{-1}\left(\norm{\begin{bmatrix}
		 		x(\tvar_0)\\
		 		\eta(\tvar_0)
		 		\end{bmatrix}}\right),\norm{\begin{bmatrix}
		 		x(\tvar_0)\\
		 		\eta(\tvar_0)
		 		\end{bmatrix}} \right\rbrace$}
		\begin{align}
		 	\nonumber	&V(x(\tvar_{j})) \\ 
		 	\nonumber	\leq& e^{\frac{-\tilde \epsilon t_{\min}}{2m} \left( j+\frac{\svar_j}{\svar_{\max}} - 2m\right)} \max \left\lbrace \bar{\alpha}_V(\norm{x(s_0)}),\norm{\eta(s_0)} \right\rbrace\\
		 	\nonumber \leq & e^{\frac{-\tilde \epsilon t_{\min}}{2m} \left( j+\frac{\svar_j}{\svar_{\max}} - 2m\right)} \resizebox{.62\linewidth}{\mylen}{$\max \left\lbrace \bar{\alpha}_V\left(\norm{\begin{bmatrix}
		 		x(\tvar_0)\\
		 		\eta(\tvar_0)
		 		\end{bmatrix}}\right),\norm{\begin{bmatrix}
		 		x(\tvar_0)\\
		 		\eta(\tvar_0)
		 		\end{bmatrix}} \right\rbrace$}\\
		 	\eqqcolon& \tilde{\beta}_1 \left(\norm{\begin{bmatrix}
		 		x(\tvar_0)\\
		 		\eta(\tvar_0)
		 		\end{bmatrix}},\svar_j,j\right). \label{eq_Vk_bound_theo2}
		 \end{align}
		 Note that $\tilde{\beta}_1$ is a class $\mathcal{K}\mathcal{L}\mathcal{L}$ function. Moreover,
		 \eqref{eq_Vk_bound_theo2} can as well be used to show a similar bound for $\eta$. In particular, we obtain due to the update of $\eta$ according to \eqref{eq_S_window} that
		 \begin{equation}
		 \label{eq_eta_bound_k}
		 	\eta(\tvar_j) \leq (m-1) \tilde{\beta}_1\left(\norm{\begin{bmatrix}
		 		x(\tvar_0)\\
		 		\eta(\tvar_0)
		 		\end{bmatrix}},t_j,j\right).
		 \end{equation} 
		To show asymptotic stability, we will now derive an analogous bound between the sampling instants. For $\eta(\svar,j)$, since $\dot{\eta} = 0$ during flows, we obtain for $\svar_{j}\leq \svar \leq \svar_{j+1}$ and all $j \in \mathbb{N}_0$
		\begin{equation}
			\label{eq_eta_bound}
			\eta(\svar,j+1) \leq \eta(\tvar_j).
		\end{equation} 
		
		For $x(\svar,j)$ and $e(\svar,j)$, we can use Proposition~\ref{prop_hybrid} to derive a bound. Since, for all $j\in\mathbb{N}_0$, $\auxvar(\tvar_j^+) \leq \delta T_{\max}(\gamma_i,L_i+\frac{\epsilon_i}{2})$ for some $i \in \left[1,\dots,n_\tpar\right]$,
		 \begin{equation}
		 \label{eq_ui_bound_theo}
		 	U_i(\xi(t,j+1)) \leq e^{\max \left\lbrace \epsilon_i,1 \right\rbrace t_{\max}} 	V(x(\tvar_{j})),
		 \end{equation}
		 for $\svar_{j} \leq t < \svar_{j+1}$, where $U_i(\xi)$ is the function according to \eqref{eq_def_u} for the parameters $\gamma_i$ and $\Lambda_i = L_i +\frac{\epsilon_i}{2}$. Since $\phi \in \left[ \lambda,\lambda^{-1}\right]$ for some $\lambda > 0$, there exist two functions $\varphi_{i,1}, \varphi_{i,2} \in \mathcal{K}_\infty$, such that for all $x,e,\tau$, 
		 	$\varphi_{i,1}\left(\norm{\begin{bmatrix}
			 	x\\
			 	e
		 		\end{bmatrix}}\right) \leq U_i(\xi) \leq 	\varphi_{i,2}\left(\norm{\begin{bmatrix}
		 		x\\
		 		e
		 		\end{bmatrix}}\right)$
		 holds. This implies together with \eqref{eq_ui_bound_theo} that for $\svar_j \leq \svar \leq \svar_{j+1}$ with $\epsilon_{\max} = t_{\max} \max\left\lbrace 1, \epsilon_1,\dots,\epsilon_{n_\tpar} \right\rbrace$ and $\varphi_{\max}(\cdot) = \max \left\lbrace \varphi_{1,1}^{-1} (\cdot), \dots,\varphi_{n_\tpar,1}^{-1} (\cdot) \right\rbrace\in\mathcal{K}_\infty,$
		 \newlength\mylent
		 \settoheight\mylent{$	\norm{\begin{bmatrix}
		 		x(\svar,j)\\
		 		e(\svar,j)
		 		\end{bmatrix}} \leq \varphi_{\max}\left(e^{\epsilon_{\max}} V(x(\tvar_j))\right)
		  \leq \tilde{\beta}_2 \left(\norm{\begin{bmatrix}
		 		x(\tvar_0)\\
		 		\eta(\tvar_0)
		 		\end{bmatrix}},\svar_j,j\right),$}		 
		 \begin{equation}
		 	\resizebox{\columnwidth}{\mylent}{$\norm{\begin{bmatrix}
		 		x(\svar,j)\\
		 		e(\svar,j)
		 		\end{bmatrix}} \leq \varphi_{\max}\left(e^{\epsilon_{\max}} V(x(\tvar_j))\right)
		 	 \leq \tilde{\beta}_2 \left(\norm{\begin{bmatrix}
		 		x(\tvar_0)\\
		 		\eta(\tvar_0)
		 		\end{bmatrix}},\svar_j,j\right)$}\label{eq_xe_bound_k}
		 \end{equation}
		 holds, where $\tilde{\beta}_2 \in \mathcal{K}\mathcal{L}\mathcal{L}$.  Putting now together \eqref{eq_eta_bound_k},\eqref{eq_eta_bound} and \eqref{eq_xe_bound_k}, we obtain for $\svar_j \leq \svar \leq \svar_{j+1}$ that
		 \begin{align*}
		 	&\norm{\begin{bmatrix}
		 		x(\svar,j)\\
		 		e(\svar,j)\\
		 		\eta(\svar,j)
		 		\end{bmatrix}} \leq \norm{\begin{bmatrix}
		 		x(\svar,j)\\
		 		e(\svar,j)\\
		 		\end{bmatrix}} + \norm{\begin{bmatrix}
		 		\eta(\svar,j)
		 		\end{bmatrix}}\\
		 	 \leq& \tilde{\beta}_2 \left(\norm{\begin{bmatrix}
		 	 	x(\tvar_0)\\
		 	 	\eta(\tvar_0)
		 	 	\end{bmatrix}},\svar_j,j\right) 
		 	 + (m-1) \tilde{\beta}_1\left(\norm{\begin{bmatrix}
		 	 	x(\tvar_0)\\
		 	 	\eta(\tvar_0)
		 	 	\end{bmatrix}},\svar_j,j\right).
		 \end{align*}
		 Using a time-shift, we can thus conclude that \eqref{eq_stab_bound} holds for
		 \begin{align*}
		 	&\beta\left(\norm{\begin{bmatrix}
		 		x(\tvar_0)\\
		 		e(\tvar_0)\\
		 		\eta(\tvar_0)
		 		\end{bmatrix}},t,j\right)  \\
		 	\coloneqq&\tilde{\beta}_2 \left(\norm{\begin{bmatrix}
		 		x(\tvar_0)\\
		 		\eta(\tvar_0)
		 		\end{bmatrix}},\max\left\lbrace\svar-t_{\max},0\right\rbrace,j\right) \\
		 	+& (m-1) \tilde{\beta}_1\left(\norm{\begin{bmatrix}
		 		x(\tvar_0)\\
		 		\eta(\tvar_0)
		 		\end{bmatrix}},\max\left\lbrace\svar-t_{\max},0\right\rbrace,j\right).&\pushright{\qed}
		 \end{align*}

\end{document}